\documentclass[11pt,oneside,reqno]{elsarticle}
\usepackage[utf8]{inputenc}
\usepackage{braket}
\usepackage{slashed}
\usepackage{graphicx}
\usepackage{amsmath}
\usepackage{amssymb}
\usepackage{amsthm}
\usepackage{marginnote}
\usepackage{xcolor}
\usepackage{hyperref}
\usepackage{multicol}
\usepackage{float}
\usepackage[margin=0.75in]{geometry}
\hypersetup{
    colorlinks=true,
    linkcolor=blue,
    filecolor=magenta,      
    urlcolor=blue,
    citecolor=red,
    pdftitle={graphdensity},
    pdfpagemode=FullScreen,
    }
\graphicspath{ {./images/} }
\newtheorem{theorem}{Theorem}

\newtheorem{corollary}{Corollary}

\newtheorem{proposition}{Proposition}

\DeclareMathOperator{\tr}{Tr}

\title{Quantum Entanglement \& Purity Testing: A Graph Zeta Function Perspective}

\begin{document}
\author[1]{Zachary P. Bradshaw\corref{cor1}}
\ead{zbradshaw@tulane.edu}
\author[2]{Margarite L. LaBorde}
\cortext[cor1]{Corresponding author}
\affiliation[1]{organization={Tulane University, Department of Mathematics},
            city={New Orleans},
            country={USA}}
\affiliation[2]{organization={Louisiana State University, Department of Physics \& Astronomy},
            city={Baton Rouge},
            country={USA}}

\begin{abstract}
    We assign an arbitrary density matrix to a weighted graph and associate to it a graph zeta function that is both a generalization of the Ihara zeta function and a special case of the edge zeta function. We show that a recently developed bipartite pure state separability algorithm based on the symmetric group is equivalent to the condition that the coefficients in the exponential expansion of this zeta function are unity. Moreover, there is a one-to-one correspondence between the nonzero eigenvalues of a density matrix and the singularities of its zeta function. Several examples are given to illustrate these findings.
\end{abstract}
\maketitle

\begin{multicols}{2}
\section{Introduction}
One of the most interesting and often discussed properties arising in quantum information theory is that of quantum entanglement \cite{einstein1935}, wherein a bipartite quantum system is described by a joint state in such a way that the state of one subsystem cannot be described independently of the other, no matter the physical distance between them. In the restricted case of pure states that we consider here, a state $\ket{\psi}$ is called entangled if it cannot be written as a product state $\ket{\phi}\otimes\ket{\chi}$. States which do not possess this property are called separable. In recent decades, a number of criteria for separability have been developed, including the positive partial transpose (PPT) criterion \cite{HHH96,Per96} and $k$-extendibility \cite{DPS02,W89a}, and the problem of determining whether a state is separable or entangled has been shown to be NP-hard in many cases \cite{gharibian2008}. Additionally, quantum algorithms which test for separability, such as the SWAP test \cite{barenco1997stabilization} and its generalizations \cite{bradshaw2022}, are under continuous development.

Meanwhile, another perspective has arisen in which quantum properties are framed in a graph-theorectic setting. This study began with the work of Braunstein, Ghosh, and Severini \cite{braunstein2004}, who defined the density matrix of a graph as the normalized Laplacian associated to it. In their work, they give a graph-theorectic criterion for the entanglement of the associated density matrix; however, not all density matrices can be encoded into a graph in this way, thus limiting the field of applicability of this criterion. The work of Hassan and Joag \cite{hassan2007} addresses this problem by associating an arbitrary density matrix to a weighted graph. They then define a modified tensor product of graphs in such a way that an arbitrary quantum state is a product state if and only if it is the density matrix of a modified tensor product of weighted graphs.

The perspective of Braunstein et al. produces a graph-theorectic separability criterion known as the degree criterion \cite{braunstein2006}, which was shown to be equivalent to the PPT criterion in \cite{hildebrand2008}. In the limited case of density matrices which can be associated to a graph in this sense, the PPT criterion can be replaced by a simple graph-theorectical criterion. Here we will take another step in this direction by establishing a connection between a family of bipartite pure state separability tests, the cycle index polynomials of the symmetric group, and a new graph zeta function associated to a density matrix. The crux of the argument is that this zeta function is the generating function for the cycle index polynomial of the symmetric group evaluated at the moments of the density matrix, and this is exactly the acceptance probability of these tests. Moreover, we show that there is a one-to-one correspondence between the singularities of this zeta function and the nonzero eigenvalues of the associated density matrix, thereby establishing a variant of the Hilbert-P\'olya conjecture associated to this zeta function.

In Section~\ref{sec:sep-tests}, we review the separability tests defined in \cite{bradshaw2022,margo} related to previous work in \cite{doherty04separability,doherty05multipartite,W89a}. We then review previous constructions of the density matrix of a graph but ultimately take a different approach by naturally assigning an arbitrary density matrix to a weighted graph in Section~\ref{sec:graph-density}. The relevant graph zeta function is then introduced in Section~\ref{sec:graph-zeta}, and in Section~\ref{sec:equivalence}, we prove that our separability tests are equivalent to the condition that the expansion of the corresponding zeta function has unit coefficients. This establishes a simple graph-theorectic test for the entanglement of pure bipartite states. We also derive a determinant representation for this function, and use it to prove the correspondence between its singularities and the nonzero eigenvalues of the associated density matrix in Section~\ref{sec:sing-eigen}. Finally, we give concluding remarks in Section~\ref{sec:conclusion}.

\section{Review of Separability Tests}\label{sec:sep-tests} The separability tests outlined in \cite{bradshaw2022,margo} are examples of $G$-Bose symmetry tests, where $G$ is some finite group. Let $U:G\to \mathcal{U}(\mathcal{H})$ be a unitary representation of $G$ on the Hilbert space $\mathcal{H}$. A state $\rho$ is called $G$-Bose symmetric if
\begin{align*}
    \Pi_G\rho\Pi_G^\dagger=\rho,
\end{align*}
where $\Pi_G:=\frac{1}{\lvert G\rvert}\sum_{g\in G}U(g)$ is the projection onto the $G$-symmetric subspace, or the space of states $\ket{\psi}$ such that $U(g)\ket{\psi}=\ket{\psi}$ for all $g\in G$. To be clear, in general both mixed and pure states may demonstrate this property. In both \cite[Chapter~8]{harrow2005applications} and \cite{margo}, it was shown how to test a state for $G$-Bose symmetry using a quantum computer, and we review a special case of this procedure here.

\begin{figure}[H]
\begin{center}
\includegraphics[
width=3.25in
]{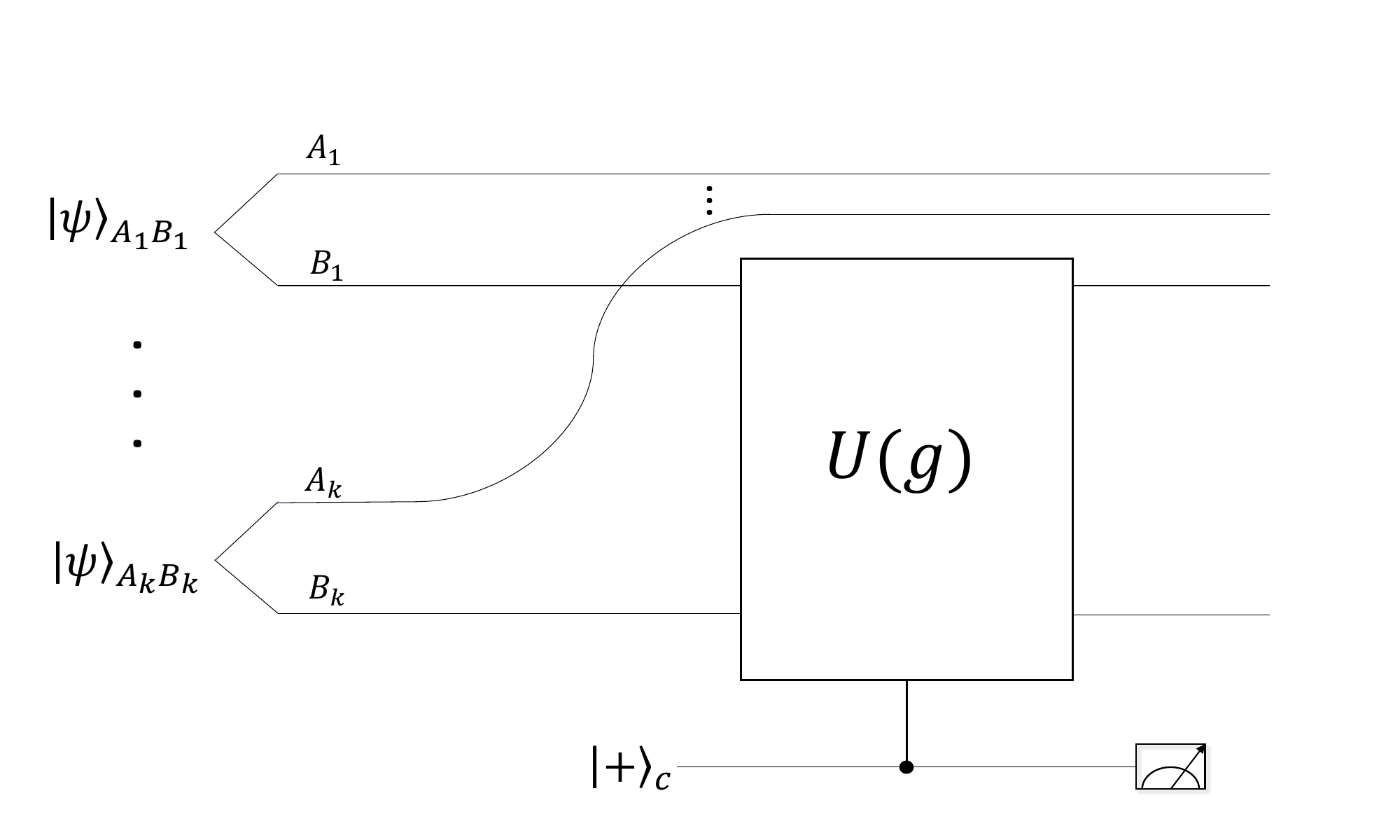}
\end{center}
\caption{Quantum circuit to implement a $G$-Bose symmetry test.}
\label{fig:circuit}
\end{figure}

The $G$-Bose symmetry property is equivalent to the condition $\left\Vert\Pi^G_S\ket{\psi_S}\right\Vert_2=1$. To test for this, we generate a superposition control state
\begin{equation*}
    \ket{+}_C = \frac{1}{\sqrt{|G|}} \sum_{g\in G} \ket{g}  \,
\end{equation*}
wherein we take a superposition over some set of computational basis elements labelled with a corresponding group element $G$, and this can be done in a general sense with a quantum Fourier transform, as discussed in \cite{margo}. This control state is used to implement a corresponding unitary $U(g)$ if the control qubit is in the state $\ket{g}$. Afterwards, applying an inverse quantum Fourier transform and measuring all control qubits, accept if the outcome $\ket{0}\!\!\bra{0}_C$ occurs and reject otherwise. The acceptance probability, given a pure state, is then given by
\begin{align*}
    \bigg\lVert(\bra{+}_C\otimes I_S)&\bigg(\frac{1}{\sqrt{|G|}}\sum_{g\in G}\ket{g}_C\otimes(U_S(g)\ket{\psi}_S)\bigg)\bigg\rVert_2^2\\
    &=\bigg\lVert\frac{1}{|G|}\sum_{g\in G}U_S(g)\ket{\psi}_S\bigg\rVert_2^2\\
    &=\left \|\Pi_S^G\ket{\psi}_S\right \|_2^2\\
    &=\tr[\Pi_S^G\ket{\psi}\!\!\bra{\psi}_S].
\end{align*}
This result is easily generalized to mixed states by convexity, and we find that the acceptance probability is equal to $\tr[\Pi_S^G\rho_S]$. 

Consider the circuit in Fig~\ref{fig:circuit}. To test for the separability of a bipartite pure state $\psi_{AB}$, we prepare $k$ copies of the state $\psi_{AB}$ and label the $k$ copies of the $A$ system by $A_1 \cdots A_k$ and the $k$ copies of the $B$ system by $B_1 \cdots B_k$. We then perform an $S_k$-Bose symmetry test on the state $\psi_{AB}^{\otimes k}$, wherein we identify $S$ with $A_1B_1\cdots A_kB_k$ and $U_S(\pi)$ with $I_{A_1\cdots A_k}\otimes W_{B_1\cdots B_k}(\pi)$, where $\pi\in S_k$ and $W_{B_1\cdots B_k}:S_k\to \mathcal{U}(\mathcal{H}_{B_1\cdots B_k})$ is the standard unitary representation of $S_k$ which acts on $\mathcal{H}_{B_1\cdots B_k} \equiv \mathcal{H}_{B_1} \otimes \cdots \otimes \mathcal{H}_{B_k}$ by permuting the Hilbert spaces according to the corresponding permutation. Define $\rho_B:= \tr_A[\psi_{AB}]$. The acceptance probability for the bipartite pure-state separability algorithm is given by
\begin{align*}
p^{(k)} := \tr[\Pi_{B_1\cdots B_k}\rho_B^{\otimes k}]
\end{align*}
where $\Pi_{B_1\cdots B_k}:= \frac{1}{k!}\sum_{\pi\in S_k}W_{B_1\cdots B_k}(\pi)$ is the projection onto the symmetric subspace. The acceptance probability $p^{(k)} = 1$ for all $k$ if and only if $\psi_{AB}$ is separable. Thus, the $S_k$-Bose symmetry test is indeed a separability test for pure bipartite states.

In \cite{bradshaw2022}, this test is generalized to any finite group $G$. Moreover, the acceptance probability of this generalized test is shown to be given by the cycle index polynomial of $G$, which is defined by $\frac{1}{\lvert G\rvert}\sum_{g\in G}x_1^{c_1(g)}\cdots x_k^{c_k(g)}$ for permutation groups, where $c_j(g)$ is the number of $j$-cycles in the cycle decomposition of $g$. This definition is easily extended to any finite group by Cayley's theorem \cite{Dummit_Foote}. When we specialize to the case $G=S_k$, this acceptance probability takes the form
\begin{align*}
p^{(k)}=\sum_{a_1+2a_2+\cdots +ka_k=k}\prod_{j=1}^k\frac{(\tr[\rho_B^j])^{a_j}}{j^{a_j}a_j!},
\end{align*}
where the sum is taken over the partitions of $k$.

\section{Density Matrices and Weighted Graphs}\label{sec:graph-density}
The study of the density matrix of a graph was initiated by Braunstein et al. in \cite{braunstein2004}. A graph $X$ is a set $V(X)=\{v_1,\ldots,v_n\}$ of labeled vertices along with a set $E(X)$ of pairs $\{v_i,v_j\}$ of vertices called edges.

The (vertex) adjacency matrix $A_X$ of the graph is given by
\begin{align*}
    (A_X)_{ij}:=
    \begin{cases} 1, & \text{if }\{v_i,v_j\}\in E(X)\\
    0, & \text{otherwise}
    \end{cases},
\end{align*}
and it encodes the information about the edges of the graph. The degree $d(v_i)$ of a vertex $v_i$ is the number of edges which include the vertex $v_i$. We define the degree matrix $\Delta_X$ of the graph to be the diagonal matrix consisting of the degrees of each vertex; that is, $\Delta_X:=\text{diag}(d(v_1),\ldots,d(v_n))$.

The combinatorial Laplacian of the graph is defined to be the difference between the degree matrix and the adjacency matrix: 
\begin{align*}
    L_X:=\Delta_X-A_X,
\end{align*}
which is both positive semi-definite and symmetric; however, it does not have unit trace. For this reason, we define the density matrix of a graph to be
\begin{align*}
    \rho_X:=\frac{\Delta_X-A_X}{\tr(\Delta_X)}.
\end{align*}

We may prescribe an arbitrary orientation to the edges so that we may label them by $e_1,\ldots,e_{\lvert E(X)\rvert}$ and their respective inverses by $e_{\lvert E(X)\rvert+1},\ldots,e_{2\lvert E(X)\rvert}$. A path in the graph $X$ is a sequence of edges $a_1\cdots a_s$ such that the origin vertex of $a_{j+1}$ is the terminal vertex of $a_j$. The primes of a graph are defined to be the equivalence classes of closed, backtrackless, tailless, primitive paths. By this we mean equivalence classes of paths such that the origin vertex of $a_1$ is the terminal vertex of $a_s$, $a_s\ne a_1^{-1}$, $a_{j+1}\ne a_j^{-1}$, and such that the path is not the power of another path. The equivalence classes are given by the cyclic shifts of $a_1\cdots a_s$. That is, we define $[a_1\cdots a_s]=[a_sa_1\cdots a_{s-1}]=\cdots=[a_2\cdots a_sa_1]$.

Clearly, there are density matrices which do not fit the graph-theoretic description outlined above. This point has been addressed by Hassan and Joag \cite{hassan2007} by generalizing the association of a density matrix to a graph to include a subset of weighted graphs. A weighted graph is a graph equipped with a weight function $\omega:E(X)\to\mathbb{C}$. In what follows, we use the notation $\omega_{ij}:=\omega(\{v_i,v_j\})$ for the weight of the edge connecting $v_i$ to $v_j$. Define the adjacency matrix of a weighted graph $(X,\omega)$ by
\begin{align*}
    (A_{X,\omega})_{ij}:=\begin{cases}
        \lvert\omega_{ij}\rvert, & \text{if }\{v_i,v_j\}\in E(X)\\
        0, & \text{otherwise}
    \end{cases},
\end{align*}
and the degree of the vertex $v_i$ by $d(v_i):=\sum_{j}\lvert\omega_{ij}\rvert$. The degree matrix is defined the same way as before. An approach similar to that in \cite{hassan2007} is to define the density matrix of a weighted graph by
\begin{align}\label{eq:density-of-graph}
    \rho_{X,\omega}:=\frac{L_{X,\omega}}{\tr(L_{X,\omega})},
\end{align}
where $L_{X,\omega}:=\Delta_{X,\omega}-A_{X,\omega}$ is the weighted Laplacian of the graph. For this to make sense, we need the weighted Laplacian to be nonzero, and graphically this is equivalent to the condition that the graph does not consist only of loops. In order to insure that the density matrix is positive semi-definite, we could restrict our attention to weight functions with magnitudes unchanged by a swap of the indices. That is, weight functions satisfying $\lvert\omega_{ij}\rvert=\lvert\omega_{ji}\rvert$. Note that this condition is satisfied by both symmetric and conjugate symmetric weight functions.

\begin{proposition}\label{prop:positive-semidef} Let $\omega$ be a weight function satisfying $\lvert\omega_{ij}\rvert=\lvert\omega_{ji}\rvert$. Then the Laplacian of the weighted graph $(X,\omega)$ is positive semi-definite.
\end{proposition}
\begin{proof}
    Let $u\in\mathbb{C}^n$ be an arbitrary vector upon which $L_{X,\omega}$ acts and denote its components by $u_i$. Then
    \begin{align}
        u^\dagger L_{X,\omega}u&=u^\dagger\Delta_{X,\omega}u-u^\dagger A_{X,\omega}u\nonumber\\
        &=\sum_id(v_i)\lvert u_i\rvert^2-\sum_{i,j}\lvert\omega_{ij}\rvert u_i^*u_j\nonumber\\
        &=\sum_{i,j}\lvert\omega_{ij}\rvert\lvert u_i\rvert^2-\sum_{i,j}\lvert\omega_{ij}\rvert u_i^*u_j\label{eq:expanding-terms}
    \end{align}
    Now the first term in \eqref{eq:expanding-terms} can be written as
    \begin{align*}
        \sum_{i}\lvert\omega_{ii}\rvert\lvert u_i\rvert^2+\sum_{i<j}\lvert\omega_{ij}\rvert\lvert u_i\rvert^2+\sum_{i<j}\lvert\omega_{ij}\rvert\lvert u_j\rvert^2.
    \end{align*}
    Similarly, the second term in \eqref{eq:expanding-terms} is given by
    \begin{align*}
        \sum_{i}\lvert\omega_{ii}\rvert\lvert u_i\rvert^2+\sum_{i<j}\lvert\omega_{ij}\rvert u_i^*u_j+\sum_{i<j}\lvert\omega_{ij}\rvert u_iu_j^*.
    \end{align*}
    Thus, we have
    \begin{align*}
        u^\dagger L_{X,\omega}u&=\sum_{i<j}\lvert\omega_{ij}\rvert(\lvert u_i\rvert^2+\lvert u_j\rvert^2-u_i^*u_j-u_iu_j^*)\\
        &=\sum_{i<j}\lvert\omega_{ij}\rvert\lvert u_i-u_j\rvert^2,
    \end{align*}
    from which it follows that the Laplacian is positive semi-definite.
\end{proof}

Since \eqref{eq:density-of-graph} obviously has unit trace, it does indeed define a density matrix. Notice, however, that the entries of this density matrix are positive real numbers. By following the proof of Proposition~\ref{prop:positive-semidef}, the reader will find that if $\lvert\omega_{ij}\rvert$ is replaced by $\omega_{ij}$ in the weighted adjacency matrix, then no conclusion can be made about the positive semi-definiteness of $\rho_{X,\omega}$. It is therefore not obvious how to extend this construction to take into account density matrices with complex coefficients. For this reason, we will take a different approach by naturally assigning an arbitrary density matrix to a weighted graph. This shortcoming is noted by Hassan and Joag in their version of the construction; although, they show that the property of having a density matrix is invariant under isomorphism of graphs. Moreover, several conditions for which a weighted graph does or does not have a density matrix are given in their work \cite{hassan2007}.

For our purposes, a simpler approach can be taken which includes all density matrices at the expense of having to exclude further graphs. Indeed, it is not hard to construct a weighted graph $(X,\omega)$ from a density matrix $\rho$. We construct it as follows:
\begin{enumerate}
    \item The number of vertices is the dimension $n$ of the Hilbert space upon which $\rho$ acts. Label them $1,\ldots,n$.
    \item If $\rho_{ij}$ is nonzero, include an edge from vertex $i$ to vertex $j$ with weight $\omega_{ij}:=\rho_{ij}$.
\end{enumerate}
Notice that in the second point, an edge is included for both $\rho_{ij}$ and $\rho_{ji}$, so that primes involving an edge from $i$ to $j$ followed by an edge from $j$ to $i$ can be considered (see Fig~\ref{fig:++graph}). It is this graph which we will use to construct our zeta function in the next section. In fact, the above procedure works for any matrix, not just density matrices. However, we will restrict our attention to those matrices which are positive semi-definite with unit trace so that relationships with quantum information can be developed. From this perspective, the only weighted graphs that we associate to a density matrix are those for which the matrix with entries $\omega_{ij}$ given by the weight function is a density matrix. While it is immediately apparent that the weights of the loops satisfy $\omega_{ii}\ge0$ with $\sum_i\omega_{ii}=1$, and the remaining weights satisfy $\omega_{ij}=\omega_{ji}^*$, a complete classification of this subset of weighted graphs remains an open question.

\section{Graph Zeta Functions}\label{sec:graph-zeta} The most famous zeta function is that of Bernhard Riemann \cite{edwards2001,patterson1988}, which is defined as the function
\begin{align*}
    \zeta(s):=\sum_{n=1}^\infty\frac{1}{n^s}
\end{align*}
for Re$(s)>1$ and its analytic continuation elsewhere. Associated to this function is the Riemann hypothesis, which states that all nontrivial zeros of the zeta function are contained on the vertical line with Re$(s)=\frac12$. Proving this hypothesis is one of the most important problems in mathematics as many results rely on the assumption of its truth. One potential approach to solving this problem is the Hilbert-P\'olya conjecture, which states that the imaginary parts of the nontrivial zeros of $\zeta$ are the eigenvalues of a self-adjoint operator.

This function was connected to the theory of prime numbers by Euler upon proving the identity
\begin{align*}
    \zeta(s)=\prod_{p\text{ prime}}(1-p^{-s})^{-1}.
\end{align*}
Subsequently, several analogous functions have been defined, including the Ihara zeta function \cite{czarneski,kotani2000,terras2010book} associated to a graph $X$, which is defined by
\begin{align*}
    \zeta_X(u):=\prod_{[P]}(1-u^{\nu(P)})^{-1},
\end{align*}
where the product is over the primes in $X$, and $\nu(p)$ denotes the length (number of edges) of $P$. There is a determinant formula for this function which can be traced back to Bass \cite{bass1992} and Hashimoto \cite{hashimoto1989} given by
\begin{align*}
    \zeta_X(u)^{-1}=(1-u^2)^{r-1}\det(I-Au+(\Delta-I)u^2),
\end{align*}
where $r=\lvert E(X)\rvert-\lvert V(X)\rvert+1$ is the rank of the fundamental group of the graph.

A generalization of the Ihara zeta function called the edge zeta function can be defined as follows: For a graph $X$, there is an associated edge matrix $W$ with entries given by the variable $\omega_{ab}$ if the terminal vertex of edge $a$ is the origin vertex of edge $b$ and $b\ne a^{-1}$, and zero otherwise. The edge zeta function is defined by
\begin{align*}
    \zeta_E(W,X)=\prod_{[P]}(1-\tilde N_E(P))^{-1},
\end{align*}
where $\tilde N_E(C)=\omega_{a_1a_2}\omega_{a_2a_3}\cdots\omega_{a_sa_1}$ is the edge norm and $C=a_1\cdots a_s$ is a closed path in $X$.

One can formulate analogs to the prime number theorem and the Riemann hypothesis for these functions, making them interesting in their own right. Here we will define a different but related zeta function $\zeta_\rho(u)$ associated to a density matrix. Indeed, we define
\begin{align}\label{eq:density-zeta-def}
    \zeta_\rho(u)=\prod_{[P]}(1-N_E(P)u^{\nu(P)})^{-1},
\end{align}
where the product is over equivalence classes of primes in the weighted graph associated to $\rho$, $\nu(P)$ again denotes the length of the prime representative $P$, and $N_E(P)$ is the product of the weights of the edges which make up $P$. If the weights of the graph were all unity, then the Ihara zeta function would be recovered. Of course, this cannot be true for a density matrix with rank greater than one since its trace is unity. However, this point is worth noting when \eqref{eq:density-zeta-def} is extended to all matrices, in which case the condition can certainly hold. 

Note that while \eqref{eq:density-zeta-def} is a generalization of the Ihara zeta function, it differs from the natural generalization seen in the book by Terras \cite{terras2010book}. However, by setting $\omega_{ab}=\rho_{ij}u$ where $i$ denotes the origin vertex of edge $a$ and $j$ denotes the origin vertex of edge $b$ (the terminal vertex of $a$), the density matrix zeta function is recovered. Indeed, the edge norm becomes the product of the weights of the closed path in the weighted graph associated to the density matrix multiplied by an extra factor of $u^{\nu(P)}$. Therefore, the density matrix zeta function defined by \eqref{eq:density-zeta-def} is a special case of the edge zeta function.

In the next section, we will show that the separability tests in \cite{bradshaw2022} and \cite{margo} corresponding to the symmetric group are equivalent to the conditions
\begin{align*}
    \frac{1}{n!}\left[\frac{d^n}{du^n}\zeta_\rho(u)\right]_{u=0}=1
\end{align*}
for the zeta function of the graph associated to the reduced density matrix of a pure state.

\section{Equivalence of Test with Zeta Function Criterion}\label{sec:equivalence}
Our main result relies on the fact that the zeta function associated to a density matrix is the generating function for the cycle index polynomial of the symmetric group evaluated at $x_j=\tr[\rho^j]$ for $j=1,\ldots,n$. To see this, we will derive an exponential expression for this function from \eqref{eq:density-zeta-def}. We will also prove the following theorem, which gives a determinate formula for this zeta function. Throughout, we will assume that $\lvert u\rvert$ is small enough to force the convergence of the series involved. The proof is similar to that in \cite{horton2006} for the edge zeta function.
\begin{theorem}\label{thm:det-thm}Let $\rho$ be a density matrix and let $\zeta_\rho$ denote the associated zeta function. Then
    \begin{align}\label{eq:det-formula}
        \zeta_{\rho}(u)=\det(I-u\rho)^{-1}.
    \end{align}
\end{theorem}
\begin{proof} Starting from the definition \eqref{eq:density-zeta-def}, we have
\begin{align*}
    \zeta_\rho(u)&=\prod_{[P]}(1-N_E(P)u^{\nu(P)})^{-1}.
\end{align*}
Now taking the logarithm, this becomes
\begin{align*}
    \log(\zeta_\rho(u))&=-\sum_{[P]}\log(1-N_E(P)u^{\nu(P)})\\
    &=\sum_{[P]}\sum_{j=1}^\infty\frac{(N_E(P)u^{\nu(P)})^j}{j}\\
    &=\sum_{j,m=1}^\infty\sum_{\substack{P \\ \nu(P)=m}}\frac{(N_E(P)u^{\nu(P)})^j}{jm},
\end{align*}
where in the second equality we have expanded the logarithm, and in the third, the inner sum is now over all prime paths of length $m$ (the equivalence class of such a prime contains $m$ elements, explaining the factor of $m$ in the denominator). Let us define an operator by $L:=\sum_{k,l}\rho_{kl}\frac{\partial}{\partial\rho_{kl}}$ and observe that $N_E(P)=\rho_{i_1j_1}\cdots\rho_{i_{\nu(P)}j_{\nu(P)}}$ for $i_1,\ldots,i_{\nu(P)},j_1,\ldots,j_{\nu(P)}$ the vertices in $P$. Then $L(N_E(P))^j=j\nu(P)(N_E(P))^j$, so that we have
\begin{align}
    L\log(\zeta_\rho(u))&=\sum_{j,m=1}^\infty\sum_{\substack{P \\ \nu(P)=m}}(N_E(P)u^{\nu(P)})^j\nonumber\\
    &=\sum_{j,m=1}^\infty\sum_{\substack{P \\ \nu(P)=m}}N_E(P^j)u^{j\nu(P)}\label{eq:sum-primes},
\end{align}
where we have used the fact that $(N_E(P))^j=N_E(P^j)$ since $P^j$ is just $j$ copies of each edge in $P$. Now, we are doing nothing more than summing over all the primes of any given length and of any given power. Thus, \eqref{eq:sum-primes} is equivalent to a sum over all closed, backtrackless, tailless paths. That is, we can drop the primitive assumption and write \eqref{eq:sum-primes} as
\begin{align*}
    L\log(\zeta_\rho(u))=\sum_{C}N_E(C)u^{\nu(C)},
\end{align*}
where the sum is over the paths mentioned above. The $ij$-th component of the $m$-th power of $u\rho$ is given by
\begin{align*}
    (u^m\rho^m)_{ij}=u^m\sum_{i_2,\ldots,i_{m}}\rho_{ii_2}\rho_{i_2i_3}\cdots\rho_{i_{m-1}i_m}\rho_{i_mj},
\end{align*}
and now we recognize the summand as the value of $N_E(C)$ for some path $C=e_{ii_2}e_{i_2i_3}\cdots e_{i_{m-1}i_m}e_{i_mj}$. Letting $i=j$, this becomes a closed, backtrackless, tailless path; therefore, we have shown
\begin{align*}
    (u^m\rho^m)_{ii}=u^m\sum_{\substack{C\\ o(C)=i\\\nu(C)=m}}N_E(C),
\end{align*}
where $o(C)$ denotes the origin vertex of $C$. Now summing over $i$, we have
\begin{align*}
    \tr[\rho^m]u^m=u^m\sum_{\substack{C\\\nu(C)=m}}N_E(C),
\end{align*}
so that summing over $m$ produces
\begin{align*}
    L\log(\zeta_\rho(u))=\sum_{m=1}^\infty\tr[\rho^m]u^m.
\end{align*}
On the other hand,
\begin{align*}
    L&\log(\det(I-u\rho)^{-1})\\
    &=L\log(\det(\exp(\log(I-u\rho)^{-1})))\\
    &=L\log(\exp(-\tr(\log(I-u\rho))))\\
    &=-L\tr(\log(I-u\rho))\\
    &=L\tr\left(\sum_{m=1}^\infty\frac{\rho^m}{m}u^m\right)\\
    &=L\sum_{m=1}^\infty\frac{\tr(\rho^m)}{m}u^m\\
    &=\sum_{m=1}^\infty\tr(\rho^m)u^m,
\end{align*}
where in the last line, we have used $L\tr(\rho^m)=m\tr(\rho^m)$. The argument is similar to the one for the relation $L(N_E(P))^j=j\nu(P)N_E(P)$ used earlier. Thus, we have shown that
\begin{align*}
    L\log(\zeta_\rho(u))=L\log(\det(I-u\rho)^{-1}).
\end{align*}
To finish off the proof, we note that with $\rho_{ij}=0$ for all $i,j$, we have $\zeta_\rho(u)=1=\det(I-u\rho)^{-1}$, so that applying the method of characteristics yields
\begin{align*}
    \log\left(\frac{\zeta_\rho(u)}{\det(I-u\rho)^{-1}}\right)=0,
\end{align*}
which implies
\begin{align*}
    \zeta_\rho(u)=\det(I-u\rho)^{-1}.
\end{align*}
\end{proof}

Theorem~\ref{thm:det-thm} tells us that the reciprocals of the nonzero eigenvalues of $\rho$ coincide with the zeros of $1/\zeta_\rho$ and therefore the singularities of $\zeta_\rho$. It thereby establishes a variant of the Hilbert-P\'olya conjecture wherein, instead of the imaginary parts of the nontrivial zeros of a zeta function, we consider the singularities of the zeta function. Then the corresponding statement is that the singularities of this zeta function are given by the reciprocal eigenvalues of a self-adjoint operator, namely the matrix $\rho$. In fact, it is clear from \eqref{eq:det-formula} that $\zeta_\rho$ has no zeros; although, it does vanish asymptotically.

Note that the matrix representation of a quantum state is basis dependent, and under a change of basis, the weighted graph associated to this matrix will change too. However, the zeta function is unchanged, and this fact follows from Theorem~\ref{thm:det-thm}. Indeed, if $P$ is a change of basis matrix and $\rho'=P\rho P^{-1}$, then $\zeta_{\rho'}(u)=\det(I-u\rho')^{-1}=\det(PP^{-1}-uP\rho P^{-1})^{-1}=\det(P)^{-1}\det(I-u\rho)\det(P)=\det(I-u\rho)=\zeta_\rho(u)$. We can therefore rest assured that $\zeta_\rho(u)$ is well-defined.

\begin{corollary}\label{cor:exp-cor} Let $\rho$ be a density matrix and let $\zeta_\rho$ denote the associated zeta function. Then
\begin{align*}
    \zeta_{\rho}(u)=\exp\left(\sum_{m=1}^\infty\frac{\tr[\rho^m]}{m}u^m\right).
\end{align*}
\end{corollary}
\begin{proof}
    Starting from Theorem~\ref{thm:det-thm}, we have
\begin{align*}
    \zeta_\rho(u)&=\det(I-u\rho)^{-1}\\
    &=\det(\exp(\log(I-u\rho)^{-1}))\\
    &=\exp(-\tr(\log(I-u\rho)))\\
    &=\exp\left(\tr\left(\sum_{m=1}^\infty\frac{\rho^m}{m}u^m\right)\right)\\
    &=\exp\left(\sum_{m=1}^\infty\frac{\tr[\rho^m]}{m}u^m\right)
\end{align*}
\end{proof}

The next theorem shows that $\zeta_\rho(u)$ is the generating function for the cycle index polynomial $Z(S_n)(1,\tr[\rho^2],\ldots,\tr[\rho^n])$. This is the key to the relationship between this zeta function and the separability tests discussed in Section~\ref{sec:sep-tests}.
\begin{theorem}\label{thm:gen-thm}
    Let $\rho$ be a density matrix and let $\zeta_\rho$ denote the associated zeta function. Then $\zeta_\rho$ is the generating function for $Z(S_n)(1,\tr[\rho^2],\ldots,\tr[\rho^n])$.
\end{theorem}
\begin{proof}
By Corollary~\ref{cor:exp-cor}, we have the identity
\begin{align*}
    \zeta_{\rho}(u)=\exp\left(\sum_{m=1}^\infty\frac{\tr[\rho^m]}{m}u^m\right).
\end{align*}
Let us split up this exponential into a product and then expand. This gives us
\begin{align*}
    \zeta_{\rho}(u)&=\prod_{m=1}^\infty\exp\left(\frac{\tr[\rho^m]}{m}u^m\right)\\
    &=\prod_{m=1}^\infty\sum_{j_m=0}^\infty\frac{\tr[\rho^m]^{j_m}}{m^{j_m}j_m!}u^{mj_m}
\end{align*}
Then the $n$-th coefficient is given by the sum over all terms where the exponent $\sum_{k=1}^{\infty}kj_{k}$ of $u$ is equal to $n$, each of which is given by a partition of $n$. That is, we have that the $n$-th coefficient is
\begin{align*}
\sum_{j_{1}+\cdots+nj_{n}=n}\prod_{k=1}^{n}\frac{(\tr[\rho^k])^{j_k}}{k^{j_{k}}j_{k}!},
\end{align*}
which is exactly the cycle index polynomial $Z(S_n)(1,\tr[\rho^2],\ldots,\tr[\rho^n])$.
\end{proof}
\begin{corollary}\label{cor:deriv}
    Let $\psi_{AB}$ be a bipartite pure state and let $\rho_B:=\tr_A[\psi_{AB}]$ be the reduced density matrix. Then $\psi_{AB}$ is separable if and only if $\frac{1}{n!}\left[\frac{d^n}{du^n}\zeta_\rho(u)\right]_{u=0}=1$ for all $n\in\mathbb{Z}_{\ge0}$.
\end{corollary}
\begin{proof}
    In Section~\ref{sec:sep-tests} (and in \cite{bradshaw2022,margo}) it was shown that a bipartite pure state $\psi_{AB}$ is separable if and only if the acceptance probability of the symmetric group separability algorithm reviewed there is 1 for all $n$, and that this is equivalent to the statement that the state is separable if and only if $Z(S_n)(1,\tr[\rho^2],\ldots,\tr[\rho^n])=1$ for all $n$. By Theorem~\ref{thm:gen-thm}, the corollary then follows.
\end{proof}

By considering the density matrix zeta function in the original form \eqref{eq:density-zeta-def}, Corollary~\ref{cor:deriv} exchanges the computation and evaluation of the cycle index polynomial of $S_n$ for a graph-theorectic calculation. It follows that the tests in \cite{bradshaw2022,margo} can be viewed from a graph-theorectic perspective just as the PPT criterion was shown to be equivalent to the graph-theoretic degree criterion \cite{braunstein2006}. To illustrate this point, let us compute a few of the coefficients. For $n=0$, the condition for separability holds trivially since 
\begin{align*}
    \zeta_\rho(0)=\prod_{\substack{[P]\\\nu(P)=0}}(1-N_E(P))^{-1},
\end{align*}
but there are no primes with zero edges, so that the product evaluates to unity. For $n=1$, we use the fact that $\frac{d}{du}\zeta_\rho(u)=\zeta_\rho(u)\frac{d}{du}\log(\zeta_\rho(u))$ and note that
\begin{align*}
    \frac{d}{du}\log(\zeta_\rho(u))&=-\frac{d}{du}\sum_{[P]}\log(1-N_E(P)u^{\nu(P)})\\
    &=\sum_{[P]}\frac{N_E(P)\nu(P)u^{\nu(P)-1}}{1-N_E(P)u^{\nu(P)}}.
\end{align*}
Now letting $u=0$, the only terms that survive are those with $\nu(P)=1$, and we have
\begin{align*}
    \left[\frac{d}{du}\zeta_\rho(u)\right]_{u=0}&=\left[\zeta_\rho(u)\frac{d}{du}\log(\zeta_\rho(u))\right]_{u=0}\\
    &=\sum_{\substack{[P]\\\nu(P)=1}}N_E(P).
\end{align*}
The only primes with $\nu(P)=1$ are the loops, which correspond to the diagonal entries of the density matrix. Therefore, $[\frac{d}{du}\zeta_\rho(u)]_{u=0}$ is the sum of the diagonal entries of $\rho$; that is, $[\frac{d}{du}\zeta_\rho(u)]_{u=0}=\tr[\rho]=1$, which is consistent with the cycle index polynomial calculation.

For the $n=2$ case, the coefficient is given by evaluating $\frac12\frac{d^2}{du^2}\zeta_\rho(u)$ at $u=0$. Observe that $\frac12\frac{d^2}{du^2}\zeta_\rho(u)$ is given by
\begin{align*}
    \frac12\left(\zeta_\rho(u)\left(\frac{d}{du}\log(\zeta_\rho(u))\right)^2+\zeta_\rho(u)\frac{d^2}{du^2}\log(\zeta_\rho(u))\right)
\end{align*}
so that evaluating at $u=0$ yields
\begin{align*}
    \frac12\left(1+\left[\frac{d^2}{du^2}\log(\zeta_\rho(u))\right]_{u=0}\right),
\end{align*}
where we have used the previous results for $n=0,1$ to simplify. We obtain the graph-theorectic version of the second term from \eqref{eq:density-zeta-def} in a similar way to the $n=1$ case. Indeed, using the standard quotient rule and evaluating at $u=0$ yields
\begin{align*}
    \bigg[\frac{d^2}{du^2}&\log(\zeta_\rho(u))\bigg]_{u=0}\\
    &=2\sum_{\substack{[P]\\\nu(P)=2}}N_E(P)+\sum_{\substack{[P]\\\nu(P)=1}}(N_E(P))^2,
\end{align*}
and the $n=2$ coefficient is therefore given by
\begin{align*}
    \frac{1}{2}\left(1+2\sum_{\substack{[P]\\\nu(P)=2}}N_E(P)+\sum_{\substack{[P]\\\nu(P)=1}}(N_E(P))^2\right).
\end{align*}
Now, the corresponding cycle index polynomial calculation produces the value $\frac12(1+\tr[\rho^2])$. Therefore, it must be the case that
\begin{align*}
    \tr[\rho^2]=2\sum_{\substack{[P]\\\nu(P)=2}}N_E(P)+\sum_{\substack{[P]\\\nu(P)=1}}(N_E(P))^2,
\end{align*}
and this is easy to see when $\rho$ is written in the diagonal basis (which exists by the spectral theorem). Indeed, in this case the graph consists only of loops so that the $\nu(P)=2$ term vanishes and the $\nu(P)=1$ term is the sum of the squares of the eigenvalues of $\rho$.

Consider the density matrix
\begin{align*}
    \rho=\frac12\begin{pmatrix}
        1&1\\
        1&1
    \end{pmatrix},
\end{align*}
\begin{figure}[H]
\centering
\includegraphics[
width=3.25in
]{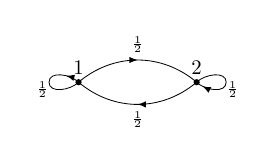}
\caption{The weighted graph associated with the pure state $\rho=\ket{+}\!\!\bra{+}$.}
\label{fig:++graph}%
\end{figure}

\noindent which is the pure state $\ket{+}\!\!\bra{+}$, so that all coefficients should be equal to unity. The weighted graph associated to $\rho$ is given in Fig~\ref{fig:++graph}. As expected, the $n=1$ coefficient is 
\begin{align*}
    \sum_{\substack{[P]\\\nu(P)=1}}N_E(P)=\frac12+\frac12=1
\end{align*}
and the $n=2$ coefficient is
\begin{align*}
    \frac{1}{2}&\left(1+2\sum_{\substack{[P]\\\nu(P)=2}}N_E(P)+\sum_{\substack{[P]\\\nu(P)=1}}(N_E(P))^2\right)\\
    &=\frac12\left(1+2\left(\frac12\cdot\frac12\right)+\left(\frac12\right)^2+\left(\frac12\right)^2\right)=1.
\end{align*}

Next consider the maximally mixed state $\rho=\frac12I$ on one qubit. Since this state is mixed, every purification of it is entangled. The Bell state $\tfrac{1}{\sqrt{2}}(\ket{00}+\ket{11})$ is such a purification since tracing out either subsystem gives the state $\frac12(\ket{0}\!\!\bra{0}+\ket{1}\!\!\bra{1})=\frac12I$. Therefore, by computing the $n=2$ coefficient of $\zeta_\rho(u)$, we get a graph-theorectic proof that the Bell state is entangled. Indeed, the weighted graph associated to $\rho$ is given in Fig~\ref{fig:maxmixgraph} and consists only of loops since the density matrix is diagonal. There are therefore no primes with $\nu(P)=2$ and the only two primes with $\nu(P)=1$ are the loops themselves. Thus, the $n=2$ coefficient is $\frac12(1+0+(\tfrac14+\tfrac14))=\frac34<1$, from which it follows that $\psi=\tfrac{1}{\sqrt{2}}(\ket{00}+\ket{11})$ is entangled. Note that the $n=2$ case is the graph-theorectic equivalent to the well-known SWAP test \cite{barenco1997stabilization,swap}.

\begin{figure}[H]
\centering
\includegraphics[
width=3.25in
]{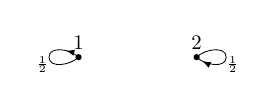}
\caption{The weighted graph associated with the maximally mixed state $\rho=\frac12I$.}
\label{fig:maxmixgraph}%
\end{figure}

It is perhaps worth noting that these tests for each $n$ double as tests for the purity of a density matrix. The $n=2$ case measures the purity exactly. In fact, this is really the mechanism behind the separability test since a pure bipartite state is separable if and only if its reduced state is pure.

\section{Singularity-Eigenvalue Correspondence}\label{sec:sing-eigen} Observe that $\det(I-u\rho)$ is a polynomial in $u$ of degree at most equal to the dimension of the Hilbert space upon which $\rho$ acts. Then the number of zeros of the reciprocal zeta function $1/\zeta_\rho$ is given by the fundamental theorem of algebra, and this coincides with the number of singularities of $\zeta_\rho$. Suppose that $\psi_{AB}$ is a pure bipartite state and let $\rho_B$ denote its reduced density matrix given by tracing out the $A$-subsystem. If $\rho_B$ is pure, it is a rank one operator such that the only nonzero eigenvalue is $\lambda=1$. Since the reciprocals of the nonzero eigenvalues of $\rho_B$ coincide with the zeros of $1/\zeta_\rho$ by Theorem~\ref{thm:det-thm}, we have another criterion for the separability of pure bipartite states.
\begin{corollary}\label{cor:singularity}
    Let $\psi_{AB}$ be a pure bipartite state and let $\rho_B$ denote its reduced density matrix by tracing out one of the systems. Then $\psi_{AB}$ is separable if and only if the only singularity of $\zeta_\rho(u)$ is at $u=1$.
\end{corollary}
\begin{proof}
    By Theorem~\ref{thm:det-thm}, we have
    \begin{align*}
        \frac{1}{\zeta_{\rho_B}(u)}=\det(I-u\rho_B).
    \end{align*}
    Notice that if $u=0$, then $\zeta_\rho\ne0$. Then the condition that $1/\zeta_\rho(u)=0$ is equivalent to the eigenvalue problem
    \begin{align*}
        \rho_B\ket{\phi}=\frac{1}{u}\ket{\phi}.
    \end{align*}
    Now, $\psi_{AB}$ is separable if and only if $\rho_B$ is a pure state, which is true if and only if the only nonzero eigenvalue is $\frac{1}{u}=1$, so that we have $u=1$. This is equivalent to the only zero of the reciprocal zeta function being at $u=1$, which is equivalent to the only singularity of $\zeta_{\rho_B}$ being at $u=1$.
\end{proof}

\begin{figure}[H]
\centering
\includegraphics[
width=3.25in
]{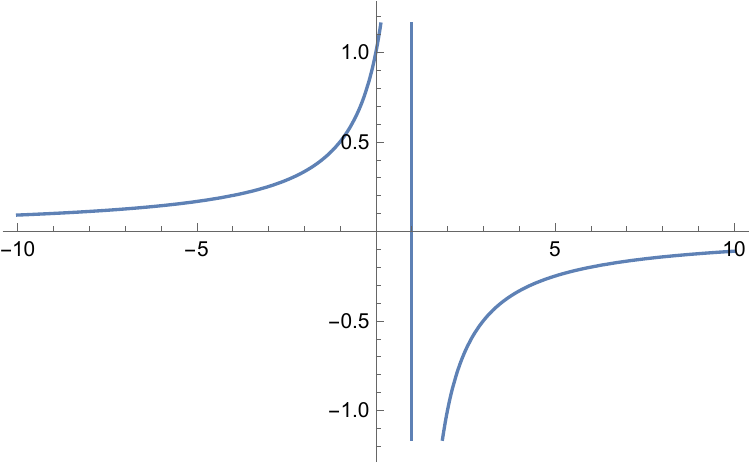}
\caption{Plot of $\zeta_\rho$ with $\rho=\ket{+}\!\!\bra{+}$.}
\label{fig:++state}%
\end{figure}

To illustrate this point, we now examine the plots of the zeta function associated to different choices of density matrices. In Fig~\ref{fig:++state}, we have an archetypal example of a pure state, $\ket{+} = \frac{1}{\sqrt{2}}(\ket{0}+\ket{1})$, and the singularity at $u=1$ is marked by a vertical line. Let us consider something a little more interesting. Since the $W$-state given by $\ket{W}=\frac{1}{\sqrt{3}}(\ket{001}+\ket{010}+\ket{100})$ is a pure state, its zeta function should look the same as that in Fig~\ref{fig:++state}. However, this state is entangled, so that we expect to pick up at least one singularity which is not at $u=1$ in the zeta function after the first system is traced out. This is exactly what we see in Fig~\ref{fig:wstate}.

\begin{figure}[H]
\centering
\includegraphics[
width=3.25in
]{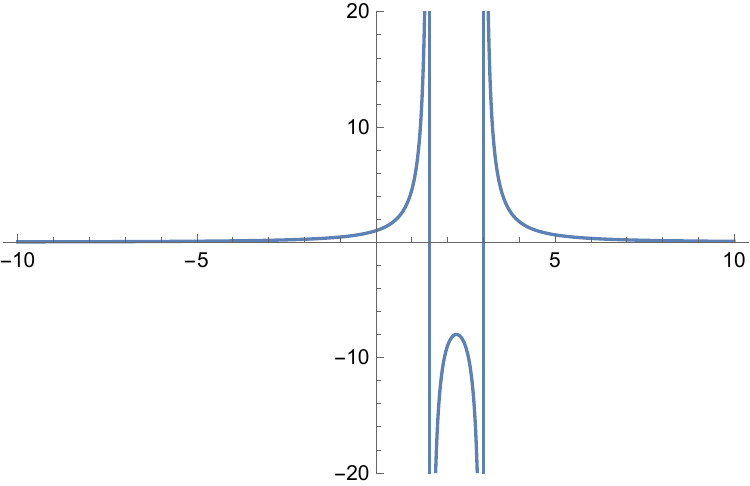}
\caption{Plot of $\zeta_\rho$ with $\rho=\tr_1[\ket{W}\!\!\bra{W}]$.}
\label{fig:wstate}%
\end{figure}

If we instead take $\psi_{AB}$ to be the GHZ-state given by $\ket{GHZ}=\frac{1}{\sqrt{2}}(\ket{000}+\ket{111})$, we find a similar result to $\ket{W}$, except that the zeta function looks somewhat different. Here it turns out that both zeros of the reciprocal zeta function of the reduced state $\rho_B$ are located at $u=2$, so that this is the only singularity that appears in Fig~\ref{fig:ghzstate}.

\begin{figure}[H]
\centering
\includegraphics[
width=3.25in
]{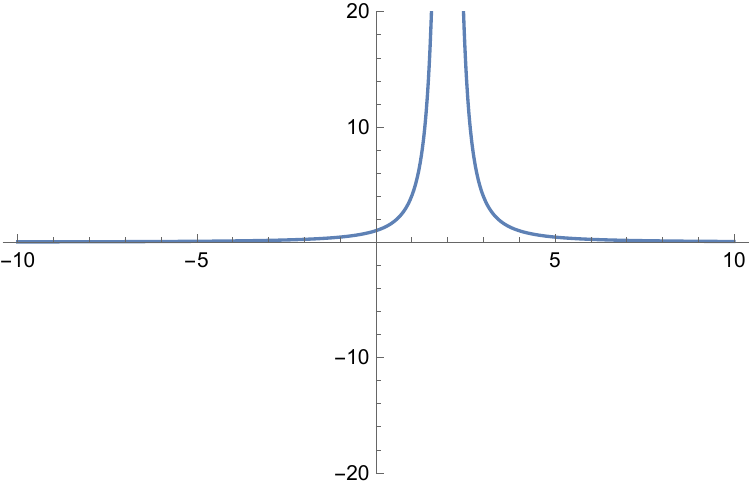}
\caption{Plot of $\zeta_\rho$ with $\rho=\tr_1[\ket{GHZ}\!\!\bra{GHZ}]$.}
\label{fig:ghzstate}%
\end{figure}

Another interesting choice for $\rho_B$ is given by the isotropic state that continuously transforms between the maximally entangled and maximally mixed states. Indeed, we define
\begin{align}\label{eq:isotropic}
    \rho(p)=\begin{pmatrix}\frac{2-p}{4}&0&0&\frac{1-p}{2}\\
    0&\frac{p}{4}&0&0\\
    0&0&\frac{p}{4}&0\\
    \frac{1-p}{2}&0&0&\frac{2-p}{4}
    \end{pmatrix},
\end{align}
from which we recover the maximally entangled state on two qubits at $p=0$ and the maximally mixed state at $p=1$. A three dimensional plot of $\zeta_\rho(u,p)$ is given in Fig~\ref{fig:3disotropic}. The maximally entangled state is a pure state, so that the cross section at $p=0$ looks like Fig~\ref{fig:++state}. The zeta function of the maximally mixed state, being the cross section at $p=1$, is shown in Fig~\ref{fig:maxmix}.

\begin{figure}[H]
\centering
\includegraphics[
width=3.25in
]{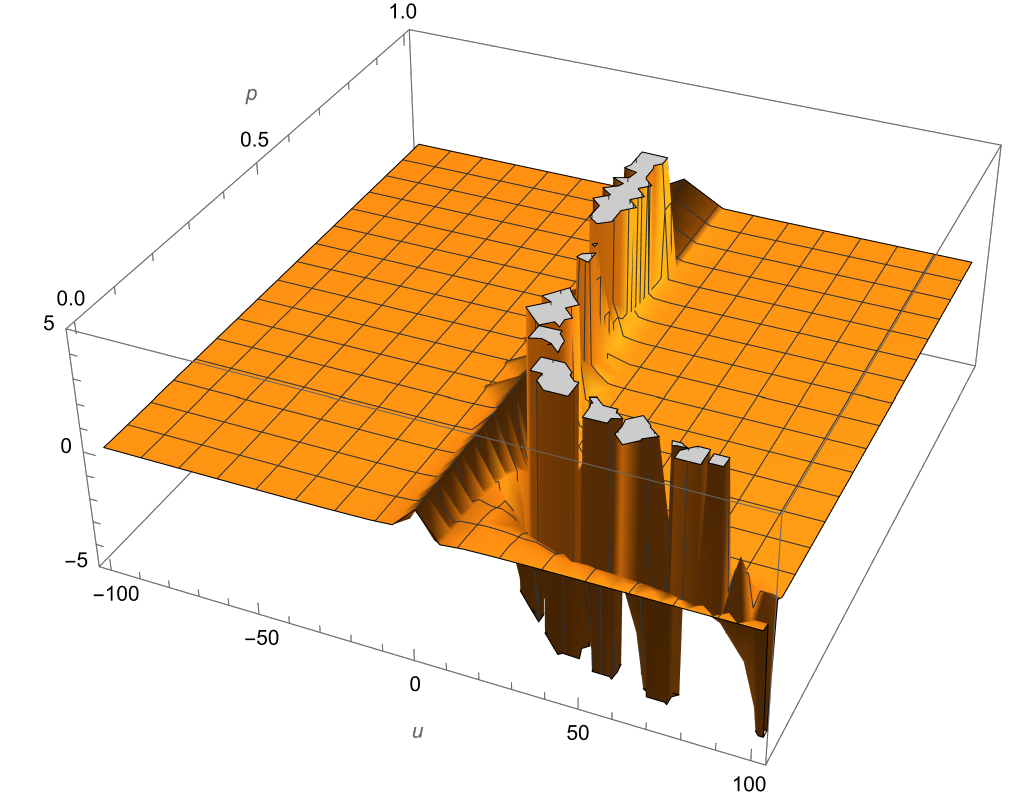}
\caption{Plot of $\zeta_{\rho}(u,p)$ with $\rho(p)$ the isotropic state defined in \eqref{eq:isotropic}.}
\label{fig:3disotropic}%
\end{figure}

\begin{figure}[H]
\centering
\includegraphics[
width=3.25in
]{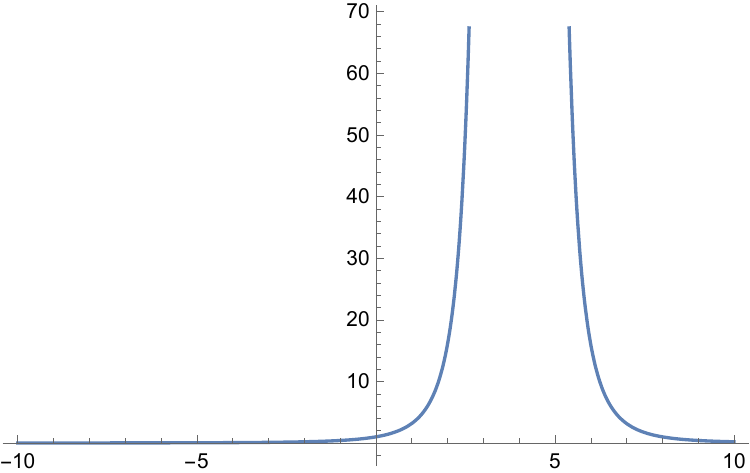}
\caption{Plot of $\zeta_\rho$ with $\rho$ the maximally mixed state on two qubits.}
\label{fig:maxmix}%
\end{figure}

We see from Fig~\ref{fig:3disotropic} that as $p$ varies from zero to one, a second pole comes from infinity and recombines with the original $u=1$ pole when $p=1$, where the state is maximally mixed; however, the merging of the poles does not happen at $u=1$, as the first pole shifts to the right as $p$ increases. Instead, the merging happens at the dimension of the Hilbert space (in this case 4), and this can be seen from \eqref{eq:isotropic} with $p=1$, where the only non-zero eigenvalue is $\frac14$.

It should be noted that neither Corollary~\ref{cor:deriv} nor Corollary~\ref{cor:singularity} apply to mixed states in general. To see this, define
\begin{align*}
    \sigma:=\begin{pmatrix}
        \frac13&0&0&0\\
        0&0&0&0\\
        0&0&0&0\\
        0&0&0&\frac23
    \end{pmatrix},
\end{align*}
and set $\rho_B:=\sigma\otimes\sigma\otimes\sigma$. Then there are singularities at $u=\frac{27}{8},\frac{27}{4},\frac{27}{2},$ and $27$, as can be seen in Fig~\ref{fig:kronecker}. These singularities correspond to the three-fold products of the nonzero reciprocal eigenvalues of $\sigma$, as expected. Moreover, the coefficients in the expansion of the zeta function go to zero as seen in Fig~\ref{fig:coefficientgraph}.

\begin{figure}[H]
\centering
\includegraphics[
width=3.25in
]{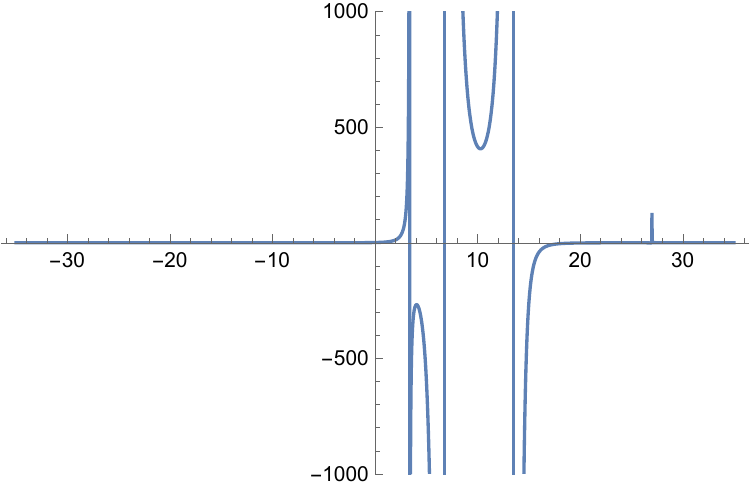}
\caption{Plot of $\zeta_\rho$ with $\rho=\sigma\otimes\sigma\otimes\sigma$.}
\label{fig:kronecker}%
\end{figure}

\begin{figure}[H]
\centering
\includegraphics[
width=3.25in
]{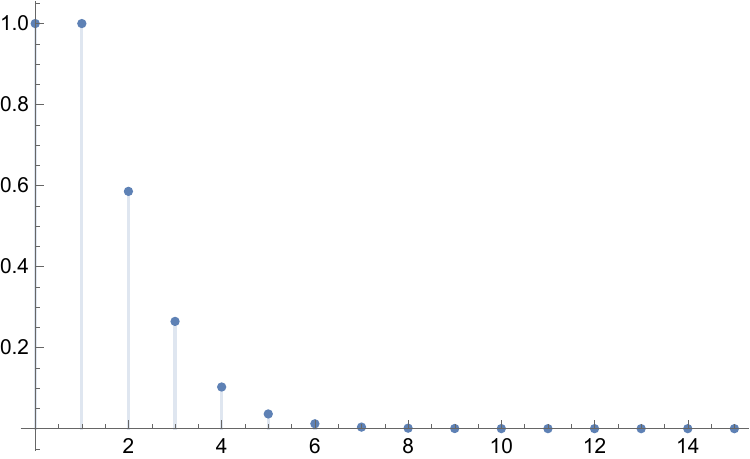}
\caption{Coefficients in expansion of $\zeta_\rho$ with $\rho=\sigma\otimes\sigma\otimes\sigma$.}
\label{fig:coefficientgraph}%
\end{figure}

By tracing out the two extra copies of $\sigma$, we recover the original state $\sigma$, which is mixed. The corresponding plots for $\zeta_\sigma$ and its coefficients are shown in Fig~\ref{fig:nonkronecker} and Fig~\ref{fig:nonkroneckerplot}, respectively. This shows that the tests we have developed for pure states do not hold for mixed states. Therefore, another method will have to be used to check for entanglement in this case.

\begin{figure}[H]
\centering
\includegraphics[
width=3.25in
]{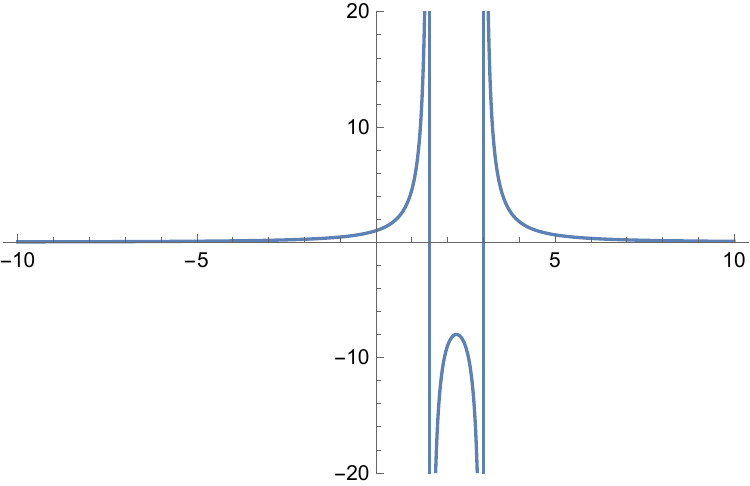}
\caption{Plot of $\zeta_\sigma$ with $\sigma=\frac13\ket{00}\!\!\bra{00}+\frac23\ket{11}\!\!\bra{11}$.}
\label{fig:nonkronecker}%
\end{figure}

\begin{figure}[H]
\centering
\includegraphics[
width=3.25in
]{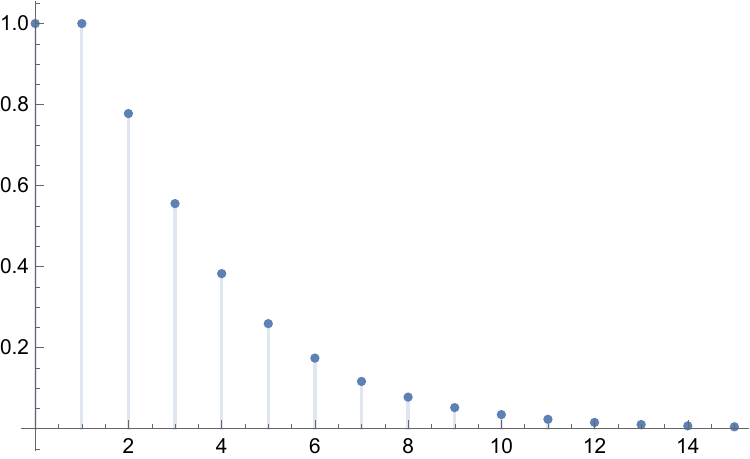}
\caption{Coefficients in expansion of $\zeta_\sigma$.}
\label{fig:nonkroneckerplot}%
\end{figure}

\section{Conclusion}\label{sec:conclusion}
A connection between the zeta function $\zeta_\rho$ defined in \eqref{eq:density-zeta-def} and the separability of bipartite pure states was derived. However, this result does not carry over to the more general case where the joint state is mixed. Establishing a similar result in the mixed state setting is worth an investigation of its own. On paper, to test whether a pure bipartite state is entangled, it suffices to compute the $n=2$ case since the reduced state has unit purity if and only if the joint state is separable. From the graph-theoretic perspective, this means that the only relevant primes when testing for entanglement are those with length $\nu(P)=1$ or $2$. The higher order tests in \cite{bradshaw2022,margo} are introduced because the acceptance probability decays as $n\to\infty$, so when noise is introduced, it may be beneficial to perform a higher order test.

Through the determinant formula \eqref{eq:det-formula}, it was shown that the nonzero eigenvalues of a density matrix $\rho$ are in correspondence with the singularities of $\zeta_\rho$. Such a connection has the potential to transform the study of the spectra of random density matrices to a graph-theorectic setting, and we leave this to future work. The reader is directed to \cite{terras2010book} for more about the connections between random matrix theory and the various graph zeta functions.

\section*{Data Availability Statement} \sloppy The Mathematica code used in this work is available in the following GitHub repository: \\ \url{https://github.com/mlabo15/ZetaFunctions}.
\section*{Acknowledgements} ZPB and MLL acknowledge support from the Department of Defense SMART scholarship.
\end{multicols}

\bibliographystyle{plainurl}
\bibliography{Ref}
\end{document}